\documentclass[runningheads,orivec]{llncs}
\usepackage[ruled]{algorithm}
\usepackage[noend]{algpseudocode}

\usepackage{graphicx} 
\usepackage{lineno}
\usepackage[T1]{fontenc}
\usepackage{amssymb} 
\usepackage{etoolbox}
\usepackage{makecell}

\makeatletter
\newcommand{\lncsqed}{%
  \ifhmode
    \qed
  \else
    \hfill\qed\par
  \fi
}
\makeatother

\let\lncsproof\proof
\let\endlncsproof\endproof
\renewenvironment{proof}{\lncsproof}{\lncsqed\endlncsproof}
\usepackage{tikz-cd}
\usepackage{array}
\usepackage{mdframed}
\usepackage{thm-restate}
\usepackage{pgf-pie}
 \usepackage{mathrsfs}
 \usepackage{enumerate}
 \usepackage{enumitem}

\usepackage{underscore}
\usepackage{pgfornament}
\usepackage{scalerel}
\usepackage{MnSymbol}
\usepackage{csquotes}
\usepackage{colortbl}
\usepackage{hyperref}
\usepackage{cleveref}
\usepackage{color}

\usepackage{afterpage} 
\usepackage{hhline}
\usepackage{multirow}
\usepackage{longtable}
\usepackage{adjustbox}

\usepackage[T1]{fontenc}
\usepackage[ttdefault=true]{AnonymousPro}

\newcommand{\Exists}[1]{\exists\,#1.\:}

\newcommand{\Forall}[1]{\forall\,#1.\:}

\newcommand{\A}{\ensuremath{\mathcal{A}}}

\newcommand{\B}{\ensuremath{\mathcal{B}}}

\newcommand{\T}{\ensuremath{\mathcal{T}}}

\newcommand{\G}{\ensuremath{gentle}}

\newcommand{\N}{\ensuremath{\mathbb{N}_{\omega}}}

\newcommand{\minmod}{\ensuremath{\mathrm{minmod}}}

\newcommand{\Teven}{\ensuremath{\T_{even}^{\infty}}}

\newcommand{\NNEQ}[1]{\ensuremath{\neq(#1_{1},\ldots,#1_{n})}}

\newcommand{\distinct}[1]{\NNEQ{x}}

\renewcommand{\int}[2]{\mathcal{#1}/\mathcal{#2}}

\newcommand{\Tle}{\ensuremath{\T_{\leq}}}

\newcommand{\Spn}{\ensuremath{\Sigma_{P}^{\mathbb{N}}}}

\newcommand{\No}{\ensuremath{\mathbb{N}^{*}}}

\newcommand{\spec}{\ensuremath{\textit{Spec}}}

\newcommand{\dom}[1]{\ensuremath{dom(#1)}}

\newcommand{\Test}[1]{\ensuremath{\T_{\texttt{#1}}}}

\newcommand{\class}[1]{\ensuremath{\mathfrak{T}_{\text{#1}}}}

\newcommand{\omg}[1]{\ensuremath{#1_{\omega}}}

\newcommand{\Testcfs}{\ensuremath{\T_{=}^{P}}}

\newcommand{\Gal}{\ensuremath{G}}

\newcommand{\unc}{\ensuremath{U}}

\begin{document}

\title{
Two generalizations of shininess
}

\author{
\begin{minipage}{0.4\textwidth}
\centering
Guilherme V. Toledo
\end{minipage}
 \and
 \begin{minipage}{0.4\textwidth}
\centering
Yoni Zohar
\end{minipage}
}
\institute{
\begin{minipage}{0.4\textwidth}
\centering
University of São Paulo, Brazil
\end{minipage}
\begin{minipage}{0.4\textwidth}
\centering
Bar-Ilan University, Israel
\end{minipage}
}
\maketitle              

\begin{abstract}
Shininess is one of the main theory combination methods in Satisfiability Modulo Theories.
We present two infinite families of generalizations of it:
one involving a property between gentleness and shininess, which helps us see gentleness itself as a variation of shininess;
and another having no further computability requirements on the weaker theory, in addition dropping the need for smoothness.
We prove these methods sharp, and compare them with existing methods.\footnote{Funding: U.S-Israel Binational Science Foundation (BSF) grant 2024049, and Israel Science Foundation (ISF) grants 619/21 and 209/25.}
\end{abstract}

\section{Introduction}

\textbf{SAT}, standing for \underline{sat}isfiability, is the area of computer science which studies the efficiency of methods that decide whether a Boolean formula is satisfiable.
Extending on this, \underline{S}atisfiability \underline{M}odulo \underline{T}heories, or \textbf{SMT} for short, automatizes the process of establishing whether a first-order formula is satisfied, now with respect to a first-order theory.
In this context, if a theory is the combination of other two, we would like to combine the algorithms of the latter into an algorithm for the former:
there is a variety of methods to accomplish this, and the sub-field in SMT that deals with this problem is called {\em theory combination.}

There are three combination methods that have become widespread.
First and foremost, the Nelson--Oppen method \cite{NelsonOppen}, requiring the two theories to be combined to be decidable and stably infinite.\footnote{All technical properties of theories, as well as other notions, will be formally reviewed in the next section.}
In a close second place, the polite method, which requires one theory to be decidable and strongly polite while the other is allowed to be simply decidable \cite{polite,JB10-TR};
it is known, thanks to \cite{CasalRasga,CasalRasga2} that, for decidable theories, strong politeness is equivalent to shininess \cite{shiny}, a property that is more difficult to implement but is easier to understand.
This is the method used by some of the most robust SMT-solvers, such as cvc5~\cite{cvc5}.
And, finally, the gentle method \cite{gentle}, which requires a theory to be gentle and the other to be decidable and have computable finite spectra \cite{CADE25}.

These methods are all {\em sharp} \cite{POPL}, e.g.:
a decidable theory that can be combined with all decidable theories must be strongly polite or, equivalently, shiny.
The machinery necessary to prove sharpness of combination methods, involving the notions of test theories and a Galois connection, also provides a reasonably effective approach to finding new combination methods.
We showcase this feature by producing two infinite families of methods closely resembling shiny combination.

The first is inspired by a method from \cite{POPL} involving the notions of $n$-decidability and $n$-shininess.
It offers a new way of understanding gentleness as a variation of shininess, helping contextualize that method.
The second comes from dropping smoothness from the definition of shininess, a quite restrictive and often undesirable requirement, which imposes no computability requisites other than decidability on the weaker theory to be combined
(much like the Nelson--Oppen and strongly polite combination methods).

We start with some basic notions necessary for the rest of the paper in \Cref{sec:prem}.
\Cref{sec:oldmethod,sec:newmethod} give two sets of generalizations of shininess.
In \Cref{sec:test} we prove that these new combination results are sharp, in the sense that they cannot be theoretically strengthened.
A brief conclusion and plans for future work can be found in \Cref{sec:conclusion}.

\section{Preliminaries}
\label{sec:prem}
We denote $\mathbb{N}\setminus\{0\}$ by $\No$,  for any $S\subseteq \mathbb{N}$ the set $S\cup\{\aleph_{0}\}$ by $S_{\omega}$, and
 $\{n\in \N : p\leq n\leq q\}$ by $[p,q]$.

\subsection{First-order logic}
\label{sec:fol}
We review standard definitions and notations regarding first-order logic. For more details, see, e.g., \cite{Enderton}.
A \textbf{signature} $\Sigma$ is a countable\footnote{A set is called {\em countable} if its cardinality is either finite or $\aleph_0$.} set of function and predicate symbols with fixed arities and containing equality.
Two signatures are \textbf{disjoint} if they share only the equality symbol;
Fixed a countably infinite set of variables, we define literals, formulas and sentences (formulas with no free variables) as usual;
let $QF(\Sigma)$ be the set of quantifier-free $\Sigma$-formulas.

A \textbf{$\Sigma$-interpretation} $\A$ is a non-empty set $\dom{\A}$ (its \textbf{domain}) equipped with functions $f^{\A}$ and predicates $P^{\A}$ over $\dom{\A}$ of the same arity as their corresponding symbols, and values $x^{\A}\in\dom{\A}$ assigned to all variables $x$.
We define the truth-value of a formula $\varphi$ in an interpretation $\A$ in the usual way, and if that value is true we write $\A\vDash\varphi$.
In \Cref{card-formulas} the formula $\neq(x_{1},\ldots,x_{n})$ states the variables $x_{1}$ through $x_{n}$ all take different values (for $n\geq 2$, when $n=1$ we can take it to be any tautology);
the formulas $\psi_{\geq n}$, $\psi_{=n}$ and $\psi_{\leq n}$ state that there are  at least, exactly, and at most $n$ elements in the domain, respectively.

\begin{figure}[t]
\begin{mdframed}
\begin{equation*}
\begin{aligned}
    \neq(x_{1},\ldots,x_{n})= \bigwedge_{i=1}^{n-1}\bigwedge_{j=i+1}^{n}\neg(x_{i}=x_{j})\\
    \\\vspace{-4mm}
    \psi_{\geq n}= \Exists{{x_1,\ldots, x_n}}\NNEQ{x}
\end{aligned}
\quad
\begin{aligned}
\psi_{\leq n}= \Exists{x_1,\ldots,x_n}\Forall{y}\bigvee_{i=1}^{n}y=x_{i}\\
\\\vspace{-4mm}
\psi_{=n}=\psi_{\geq n}\wedge\psi_{\leq n}
\end{aligned}
\end{equation*}
\end{mdframed}
\vspace{-5mm}
\caption{Cardinality formulas}
\label{card-formulas}
\vspace{-5mm}
\end{figure}

A \textbf{$\Sigma$-theory} $\T$ is a set of $\Sigma$-sentences, called its \textbf{axioms}, and
if an interpretation satisfies all of them it is said to be a \textbf{$\T$-interpretation}.
A formula is \textbf{$\T$-satisfiable} if there is at least one $\T$-interpretation that satisfies it.
A $\Sigma$-theory $\T$ is \textbf{decidable} when there is an algorithm that takes a quantifier-free $\Sigma$-formula and returns whether or not it is $\T$-satisfiable.
Given signatures $\Sigma_1$ and $\Sigma_2$, $\Sigma_{1}\sqcup\Sigma_{2}$ is their disjoint union, which is the union of their symbols after renaming them to avoid overriding;
and given a $\Sigma_{1}$-theory $\T_{1}$ and a $\Sigma_{2}$-theory $\T_{2}$, their \textbf{disjoint combination} $\T_{1}\sqcup\T_{2}$ is the $\Sigma_{1}\sqcup\Sigma_{2}$-theory with axiomatization equal to the union of $\T_{1}$ and $\T_{2}$ (after possible renaming). 

The \textbf{spectrum} $\spec_{\T}(\varphi)$ of a formula $\varphi$ in a theory $\T$ is the subset of $\N$ of those $n$ such that there exists a $\T$-interpretation
$\A$ with 
$\A\vDash\varphi$ and $|\dom{\A}|=n$.
We often make use of the L\"owenheim--Skolem theorem as stated below.
\begin{theorem}[{\cite[Theorems~2.3.4 and 2.3.7]{marker2002}}] \label{LowenheimSkolem}
    Let $\Delta$ be a set of formulas over a countable signature, and let $\kappa \ge \aleph_0$. 
    Then, $\Delta$ is satisfied by an interpretation of size $\kappa$ if and only if it is satisfied by an interpretation of size $\aleph_0$.
\end{theorem}
The following is a classical result we make use of,
obtained as an immediate corollary from compactness.



\begin{restatable}{corollary}{corofcomp}
\label{compactness2}
If $\aleph_0\notin\spec_{\T}(\phi)$,
then $\spec_{\T}(\phi)$ is finite.
\end{restatable}

\subsection{Theory combination and its Galois connection}
\label{sec:theorycombprops}

$\T$ is said to have \textbf{computable finite spectra}\footnote{Papers~\cite{CADE,FroCoS,LPAR,CADE25,POPL} include many examples for all properties discussed here.} \cite{CADE25}
when there is an algorithm that takes a quantifier-free formula $\phi$ and an $n\in\No$, returning whether $n\in\spec_{\T}(\phi)$.
A theory $\T$ is \textbf{stably infinite} \cite{NelsonOppen} if $\spec_{\T}(\phi)\neq\emptyset$, for a quantifier-free $\phi$, implies $\aleph_{0}\in\spec_{\T}(\phi)$.
A theory $\T$ is \textbf{smooth} when, for every quantifier-free $\phi$, $\T$-interpretation $\A\vDash\phi$ and $\kappa\geq|\dom{\A}|$, there is a $\T$-interpretation $\B\vDash\phi$ with $|\dom{\B}|=\kappa$.
$\T$ has the \textbf{finite model property} when, for every $\T$-satisfiable quantifier-free formula $\phi$, $\spec_{\T}(\phi)\cap\No\neq \emptyset$.
And, if $\T$ is decidable, its \textbf{minimal model function} is the function $\minmod_{\T}$ that takes a quantifier-free formula $\phi$ and returns $\min\spec_{\T}(\phi)$, if that set is not empty (and any value in $\N$ otherwise).
A theory is then said to be \textbf{shiny} \cite{shiny} when it is smooth, has the finite model property, and its minimal model function is computable.

\begin{proposition}\label{prop:alternative-definition-shininess}
    A theory is shiny if and only if there is an algorithm that takes a quantifier-free formula $\phi$ and returns an $n\in\mathbb{N}$ such that $\spec_{\T}(\phi)=\emptyset$ if $n=0$, or $\spec_{\T}(\phi)=[n,\aleph_{0}]$ otherwise.
\end{proposition}

\begin{proof}
    Suppose $\T$ is shiny:
    if $\phi$ is not $\T$-satisfiable we return $0$;
    if it is, we make $n=\minmod_{\T}(\phi)$.
    This value is in $\mathbb{N}$ because $\T$ has the finite model property;
    and $\spec_{\T}(\phi)=[n,\aleph_{0}]$ because $\T$ is smooth.
    For the converse, we make $\minmod_{\T}(\phi)=n$;
    $\T$ has the finite model property because this value is always in $\mathbb{N}$;
    and it is smooth because $\spec_{\T}(\phi)=[n,\aleph_{0}]$ (\Cref{LowenheimSkolem}).
\end{proof}

%
A theory $\T$ is \textbf{gentle} (\cite{gentle}, though we use the equivalent definition of \cite{CADE25}) when there is an algorithm that takes a quantifier-free formula $\phi$ and returns a pair $(S,b)$, where $S\subset\No$ is finite
and $b$ is a Boolean, such that:
if $b$ is true, $\spec_{\T}(\phi)=S$;
and if $b$ is false, $\spec_{\T}(\phi)=\N\setminus S$.
Using the notations for sets\footnote{These, in principle, would constitute  proper classes. But, we can select
all symbols for our signatures from a countable pool of  symbols to ensure these are sets.}  of theories found in \Cref{tab:notation-prop-th-def}(where, essentially, $\class{X}$ means the set of decidable theories with property $X$) we present the following proposition, whose proof can be found in Proposition 6.1 of \cite{POPL}.
It says that for decidable theories,
shininess implies gentleness, that in turn
implies having computable finite spectra.

\begin{table}[t]
\centering
\renewcommand{\arraystretch}{0.95}
\setlength{\tabcolsep}{6pt}

\begin{tabular}{@{}c p{5.5cm} c@{}}
\hline
\textbf{Symbol} & \textbf{Description (decidable and \ldots)} & \textbf{Definition} \\
\hline
$\class{}$
& --- 
& \multirow{5}{*}{\centering\Cref{sec:prem}} \\
$\class{shiny}$ 
& Shiny 
& \\
$\class{\G}$ 
& Gentle 
& \\
$\class{SI}$ & Stably infinite &\\
$\class{CFS}$ 
& Computable finite spectra 
& \\
\hline
$\class{S-d}$ & $S$-decidable & \multirow{2}{*}{\Cref{sec:oldmethod}}\\
$\class{S-s}$ & $S$-shiny & \\
\hline
$\class{S-r}$ & $S$-restricted & \multirow{2}{*}{\Cref{sec:newmethod}}\\
$\class{S-f}$ & $S$-final & \\
\hline
\end{tabular}
\vspace{2mm}
\caption{Notations for sets of theories.}
\label{tab:notation-prop-th-def}
\end{table}

\begin{proposition}\label{prop:known-relations-properties}
$\class{shiny}\subseteq\class{gentle}\subseteq\class{CFS}\subseteq\class{}$.
\end{proposition}

One important tool in what is to come is the following  result from \cite{gentle,POPL}  .

\begin{lemma}[Fontaine's lemma]\label{lem-fontaine}
    $\T_1 \sqcup \T_2$ is decidable iff the following problem is decidable: given conjunctions of literals $\varphi_1$ and $\varphi_2$, determine whether $\spec_{\T_1}(\varphi_1) \cap \spec_{\T_2}(\varphi_2) = \emptyset$.
\end{lemma}


Two theories are \textbf{combinable} if their disjoint union is decidable, and two sets of theories are combinable if any theory in the first is combinable with any theory in the second;
if one of these sets is a singleton, we may simply say the theory it contains is combinable with the other set.
For each $X\subseteq\class{}$, let $G(X)$ 
be the 
set of all decidable theories that are combinable with $X$, that is, the largest set of decidable theories that is combinable with $X$.

If $X\subseteq Y$, as every theory combinable with $Y$ is also combinable with $X$, we get $G(Y)\subseteq G(X)$;
and $X$ is combinable with $Y$ iff $Y$ is combinable with $X$, meaning 
$X\subseteq G(Y)$ iff
$Y\subseteq G(X)$.
Both these facts
 make $G$ an \textbf{antitone Galois connection} with respect to $\subseteq$.
A \textbf{combination theorem} typically says two sets of theories are combinable, meaning $X\subseteq G(Y)$ and $Y\subseteq G(X)$.
Such a theorem is {\em sharp} if, in fact,
$X=G(Y)$ and $Y=G(X)$.
In this case $G(G(X))=X$ and $G(G(Y))=Y$, meaning, in lattice-theoretic terms, that $X$ and $Y$ are \textbf{closed};
for an arbitrary $X$, $G(G(X))$ is its \textbf{closure}, the least closed set containing $X$.

\section{First Generalization}\label{sec:oldmethod}
We now provide a new combination method, which
generalizes shiny combination, and requires the following properties, that parameterize shininess and decidability by a computable set $S$.

\begin{definition}
\label{def:sshiny-sdec}
    Take a computable set $S\subseteq\No$.
    \begin{enumerate}
        \item A decidable theory $\T$ is $S$-decidable when there is an algorithm that takes a quantifier-free formula $\phi$ and an $n\in \No$ and returns:
        whether $n\in\spec_{\T}(\phi)$ if $n\in S$;
        and an arbitrary Boolean value otherwise.
        \item A decidable theory $\T$ is $S$-shiny if there is an algorithm that takes a quantifier-free formula $\phi$ and returns a pair $(n,S_{0})$, where $n\in\mathbb{N}$ and $S_{0}$ is a finite subset of $S$ (possibly empty), such that:
        if $n=0$, $\spec_{\T}(\phi)=S_{0}$;
        if $n\neq 0$, $\spec_{\T}(\phi)=S_{0}\cup[n,\aleph_{0}]$.
    \end{enumerate}
\end{definition}

\begin{remark}
$n$-decidability and $n$-shininess were defined in \cite{POPL} for any $n\in\No$.
The former corresponds to $\{n\}$-decidability, the latter to $\{n\}$-shininess.
\end{remark}

Now, it was shown in \cite{CADE25} that every shiny theory is decidable. This also holds for the new generalizations. 

\begin{proposition}
    An $S$-shiny theory is $S$-decidable.
\end{proposition}

\begin{proof}
    Suppose $\T$ is $S$-shiny, and that we are given a quantifier-free $\phi$ and an $m\in S$:
    we want to return whether $m\in\spec_{\T}(\phi)$.
    As $\T$ is $S$-shiny, there is an algorithm that takes $\phi$ and returns a pair $(n,S_{0})$ such that $S_{0}\subseteq S$ is finite, if $n=0$ then $\spec_{\T}(\phi)=S_{0}$, and otherwise $\spec_{\T}(\phi)=S_{0}\cup [n,\aleph_{0}]$.
    If $n=0$, $m\in\spec_{\T}(\phi)$ iff $m\in S_{0}$;
    otherwise $m\in\spec_{\T}(\phi)$ iff $m\in S_{0}$ or $m\geq n$.
    This procedure may be seen in \Cref{alg:S-shiny-implies-S-decidable}.
\end{proof}

\begin{figure}[t]
           \begin{algorithm}[H]
    \begin{algorithmic}[1]
    \Function{SD}{$\phi$,$m$}
    \State{$(n,S_{0})\gets$ computed for $\T$ and $\phi$, according to \Cref{def:sshiny-sdec}}
    \If{$m\in S_{0}$}
        \Return true
        \Else 
            \If{$n>0$}
                \If{$m\geq n$} \Return true
                \EndIf
            \EndIf
            \State{\Return false}
    \EndIf
    \EndFunction 
\end{algorithmic}
    \caption{Pseudocode for showing $\T$ is $S$-decidable if it is $S$-shiny.} \label{alg:S-shiny-implies-S-decidable}
\end{algorithm}
\end{figure}

First, we show that these notions are indeed
generalizations of decidable and shiny theories,
in particular when $S$ is taken to be the empty set.
Additionally, we show that when $S$ is taken to be the entire set $\No$,
we get computable finite spectra and gentleness.
This shows gentle and shiny theory combination are two ends of a spectrum:
gentleness is the weakest form of $S$-shininess, and shininess the strongest;
reciprocally, having computable finite spectra is the strongest form of $S$-decidability, and decidability itself the weakest.

\begin{restatable}{proposition}{propaboutSnd}
\label{propabout-S-n-d}
    $\class{$\emptyset$-d}=\class{}$, $\class{$\No$-d}=\class{CFS}$, $\class{$\No$-s}=\class{\G}$, and $\class{$\emptyset$-s}=\class{shiny}$.
\end{restatable}

\begin{proof}
    $\class{$\emptyset$-d}=\class{}$ should be clear, as no condition is imposed on its theories other than being decidable.
    If a theory $\T$ is in $\class{$\No$-d}$ then there is an algorithm that takes a quantifier-free formula $\phi$ and an $n\in\No$ and returns whether $n\in\spec_{\T}(\phi)$, what corresponds precisely to the definition of having computable finite spectra.

    That $\class{$\emptyset$-s}=\class{shiny}$ follows from the fact that in a $\emptyset$-shiny theory $\T$ there is an algorithm that takes a quantifier-free formula $\phi$ and returns a pair $(n,\emptyset)$ such that: 
    either $n=0$ and $\spec_{\T}(\phi)=\emptyset$ (meaning $\phi$ is not $\T$-satisfiable);
    or $n\in\No$ and $\spec_{\T}(\phi)=[n,\aleph_{0}]$, as expected of a shiny theory by \Cref{prop:alternative-definition-shininess}.

    Finally, it is clear $\No$-shiny theories are gentle.
    And if a theory $\T$ is gentle we either have $\spec_{\T}(\phi)$ finite, when we output $(0,\spec_{\T}(\phi))$, which is compatible with being an $\No$-shiny theory as $\spec_{\T}(\phi)\subseteq\mathbb{N}$;
    or $\spec_{\T}(\phi)$ is cofinite, and from gentleness we get $S=\N\setminus\spec_{\T}(\phi)$.
    In this case take the greatest element $n$ not in $\spec_{\T}(\phi)$, and we have $\spec_{\T}(\phi)$ equals the union of $[1,n]\setminus S$ and $[n+1,\aleph_{0}]$,
    also compatible with $\T$ being $\mathbb{N}$-shiny as certainly $[1,n]\setminus S\subseteq \mathbb{N}$, when we output $(n+1,[1,n]\setminus S)$.  
\end{proof}

\begin{figure}[t]
\begin{algorithm}[H]
    \begin{algorithmic}[1]
        \Function{SD+SS}{$\phi_1, \phi_2$}
        \If{$\phi_2$ is not $\T_2$-satisfiable}
              \Return false
            \EndIf
          \State{$(n,S_{0})\gets$ computed for $\phi_2$ and $\T_2$, according to \Cref{def:sshiny-sdec}.}
            \If{$n\neq 0$ and $\phi_1\wedge \neq(x_{1},\ldots,x_{n})$ is $\T_1$-satisfiable}
                \Comment{$x_{1},\ldots,x_{n}$ are fresh}
                \State \Return true
            \EndIf
            \For{$m\in S_{0}$} 
                \If{$m\in\spec_{\T_1}(\phi_1)$}
                    \State \Return true
                \EndIf
            \EndFor
            \State \Return false
        \EndFunction
\end{algorithmic}
\caption{Pseudocode for determining non-emptyness of $\spec_{\T_1}(\phi_1)\cap\spec_{\T_2}(\phi_2)$ when $\T_1$ is $S$-decidable and $\T_2$ is $S$-shiny.} \label{alg-sdec-sshiny}
\end{algorithm}
\end{figure}

The following theorem is the main result of this section.
It generalizes the shiny combination theorem,
according to which the disjoint combination of a shiny theory with a decidable theory is itself decidable~\cite{shiny}.

\begin{restatable}{theorem}{IpartSsSd}\label{theo:Ipart-Ss-Sd}
If $\T_{1}$ is $S$-decidable and $\T_2$ is $S$-shiny for some $S$, and both $\T_1$ and $\T_2$ are decidable, then $\T_1 \sqcup\T_2$ is decidable.   
\end{restatable}

\begin{proof}
Using \Cref{lem-fontaine},
it suffices to show that
it is then decidable whether $\spec_{\T_{1}}(\phi_{1})\cap\spec_{\T_{2}}(\phi_{2})\neq\emptyset$, where $\phi_{1}$ and $\phi_{2}$ are conjunctions of literals over the signatures $\Sigma_{1}$ and $\Sigma_{2}$, respectively.
We state \Cref{alg-sdec-sshiny} presents an adequate decision procedure.
        We can compute a pair $(n,S_{0})$, where $S_{0}\subseteq S$ is finite, such that $\spec_{\T_{2}}(\phi_{2})$ equals $S_{0}\cup [n,\aleph_{0}]$ if $n\neq 0$, and $S_{0}$ otherwise, from the fact that $\T_{2}$ is $S$-shiny.
    We can test whether $\phi_{1}\wedge\neq(x_{1},\ldots,x_{n})$ is $\T_{1}$-satisfiable from the fact that $\T_{1}$ is decidable.
    Finally, we can test if an $m\in S_{0}$ is also in $\spec_{\T_{1}}(\phi_{1})$ since $\T_{1}$ is $S$-decidable, and $S_{0}\subseteq S$.

    Next we prove correctness.\\
    ($\Rightarrow$):
    Suppose $\spec_{\T_{1}}(\phi_{1})\cap\spec_{\T_{2}}(\phi_{2})$ is not empty.
    If $n\neq 0$ and they share an element in $[n,\aleph_{0}]$ then line $4$ of our algorithm will return true.
    If not they must share an element in $S_{0}$ and the loop starting on line $5$ will  return true.\\
    ($\Leftarrow$):
    Suppose now $\spec_{\T_{1}}(\phi_{1})\cap\spec_{\T_{2}}(\phi_{2})$ is empty.
    The conditional in line $3$ won't be satisfied, as otherwise the intersection would not be empty;
    and neither will be any conditional in the loop starting at line $5$, so eventually, we will return false.
    \end{proof}

We now study how  $S$-decidability and $S$-shininess change with $S$.


\begin{restatable}{proposition}{incomparableSsSd}
\label{prop:incomparable-Ss-Sd}
~
    \begin{enumerate}
        \item If $S\subsetneq T$, then $\class{T-d}\subsetneq\class{S-d}$ and $\class{S-s}\subsetneq \class{T-s}$.
        \item If $S$ and $T$ are incomparable (i.e. $S\not\subseteq T$ and $T\not\subseteq S$), so are $\class{S-d}$ and $\class{T-d}$, and  also $\class{S-d}$ and $\class{T-d}$.
    \end{enumerate}
\end{restatable}

\begin{proof}
    Suppose $\T$ is $T$-decidable:
    this means there is an algorithm that takes a quantifier-free $\phi$ and an $n$ and, if $n\in T$, it returns whether $n\in\spec_{\T}(\phi)$;
    since $S\subseteq T$, this same algorithm shows that $\T$ is $S$-decidable.
    For $S$ and $T$-shininess it is enough to notice that if $S_{0}$ is a finite subset of $S$, and $S\subseteq T$, then $S_{0}$ is also a finite subset of $T$.
    
    Let $n_{s}\in S\setminus T$, and $n_{t}\in T\setminus S$.
    Define the $\Spn$-theory $\T_{s}$ with axiomatization
    \[\{\psi_{=n_{s}}\vee\psi_{\geq k}: k\in\No\}\cup\{\psi_{=n_{s}}\rightarrow \neg P_{k} : k\notin U\}\cup\{P_{i}\rightarrow\neg P_{j} : i\neq j\};\]
    we define $\T_{t}$ analogously.
    $\T_{s}$ is decidable:
    given a quantifier-free $\phi$, we may assume it contains at most one positive $P$-literal;
    if $\phi^{\prime}$ is obtained from $\phi$ by removing any $P$-literals, then $\phi$ is $\T_{s}$-satisfiable if and only if $\phi^{\prime}$ is $\T_{Eq}$-satisfiable.
    One direction of this is trivial;
    for the other we can use $\T_{Eq}$ is smooth to find an infinite interpretation that satisfies $\phi^{\prime}$, and we make it into a $\T_{s}$-interpretation that satisfies $\phi$ by making $P_{k}$ true if and only if it occurs positively in $\phi$.
    $\T_{s}$ is also $T$-decidable, as it contains no interpretations with cardinality in $T$;
    yet it is not $S$-decidable, as $n_{s}\in\spec_{\T_{s}}(P_{k})$ if and only if $k\in U$.
    The results for $\T_{t}$ are completely analogous.

    Now consider the theories $\T_{=n_{s}}$ and $\T_{=n_{t}}$, axiomatized by respectively $\{\psi_{=n_{s}}\}$ and $\{\psi_{=n_{t}}\}$, which are decidable (\cite{POPL}).
    $\T_{=n_{s}}$ is $S$-shiny, since $\spec_{\T_{n_{s}}}(\phi)=\{n_{s}\}$ if $\minmod_{\T_{Eq}}(\phi)\leq n_{s}$, and otherwise it is empty:
    for the non-trivial direction, if $\minmod_{\T_{Eq}}(\phi)\leq n_{s}$ we can use the fact $\T_{Eq}$ is smooth to get a $\T_{Eq}$-interpretation that satisfies $\phi$ with $n_{s}$ elements, and of course that will also be a $\T_{=n_{s}}$-interpretation.
    But $\T_{=n_{s}}$ is obviously not $T$-shiny:
    the spectrum of a tautology such as $x=x$ in $\T_{=n_{s}}$ is $\{n_{s}\}$, and this set is not of the form $T_{0}\cup[n,\aleph_{0}]$, for any $n$ and $\T_{0}\subseteq T$, since it does not contain $\aleph_{0}$;
    it is also not of the form $T_{0}$, for a finite $T_{0}\subseteq T$, as $n_{s}\notin T$.
    Analogous results hold for $\T_{=n_{t}}$.

    If $S\subsetneq T$ and $n_{t}\in T\setminus S$, the theories $\T_{t}$ and $\T_{=n_{t}}$ above also show $\class{T-d}\subsetneq\class{S-d}$ and $\class{S-s}\subsetneq\class{T-s}$.
\end{proof}

Thus, the collections of $S$-decidable and $S$-shiny theories, for computable $S\subseteq \No$, have the structure of the computable subsets of $\No$, and its dual.
More precisely, we have \Cref{fig:diagram-Ss-Sd}:
the Galois connection $G$ is simply the reflection over the dashed line.
Below it we find the many variations of $S$-shininess, having the structure of the computable subsets of $\No$, with shininess at the bottom and gentleness on top.
Above the dashed line we find the variations of $S$-decidability, having the structure of the dual of the computable subsets of $\No$, with computable finite spectra at the bottom and decidability at the top.

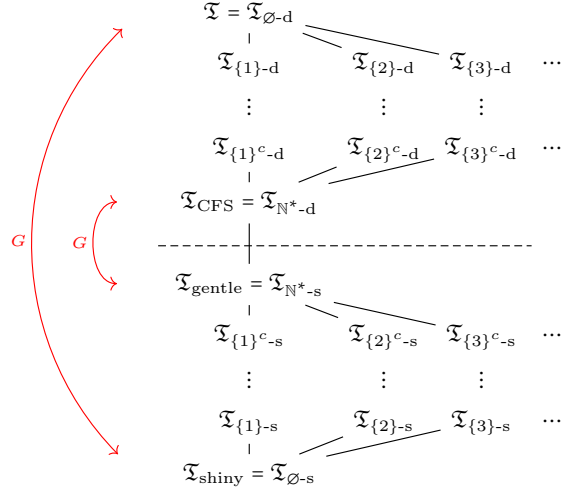
\begin{figure}[t]
    \centering
    \begin{tikzcd}[sep=0.5mm,scale=0.8]
      \phantom{O}\arrow[bend right=45,swap,red,<->]{dddddddddddddddddd}{G}& & \class{}=\class{$\emptyset$-d}\arrow[dash]{ddd}{}\arrow[dash]{dddr}{}\arrow[dash]{dddrr}{} & & & \\
      & & & & & \\
      & & & & & \\
      & & \class{$\{1\}$-d} & \class{$\{2\}$-d} & \class{$\{3\}$-d} & \cdots\\
      & & \vdots & \vdots & \vdots & \\
      & & \class{$\{1\}^{c}$-d}\arrow[dash]{ddd}{} & \class{$\{2\}^{c}$-d}\arrow[dash]{dddl}{} & \class{$\{3\}^{c}$-d}\arrow[dash]{dddll}{} & \cdots\\
      & & & & & \\
      & & & & & \\
      \phantom{O}\arrow[bend right=90,swap,red,<->]{dd}{G}& & \class{CFS}=\class{$\No$-d}\arrow[dash]{dd} & & & \\
      \phantom{O}\arrow[dashed,dash]{rrrrr}& & & & &\phantom{O} \\
      \phantom{O}& & \class{gentle}=\class{$\No$-s}\arrow[dash]{ddd}{}\arrow[dash]{dddr}{}\arrow[dash]{dddrr}{} & & & \\
      & & & & & \\
      & & & & & \\
      & & \class{$\{1\}^{c}$-s} & \class{$\{2\}^{c}$-s} & \class{$\{3\}^{c}$-s} & \cdots\\
      & & \vdots & \vdots & \vdots & \\
      & & \class{$\{1\}$-s}\arrow[dash]{ddd}{} & \class{$\{2\}$-s}\arrow[dash]{dddl}{} & \class{$\{3\}$-s}\arrow[dash]{dddll}{} & \cdots\\
      & & & & & \\
      & & & & & \\
      \phantom{O}& & \class{shiny}=\class{$\emptyset$-s} & & & \\
    \end{tikzcd}
    \caption{The rich universe of $S$-decidable and $S$-shiny theory combination.
    For readability, $\{n\}^{c}$ denotes $\No\setminus\{n\}$.}
    \label{fig:diagram-Ss-Sd}
\end{figure}

We conclude the first generalization of shininess with some examples:
\begin{example}
\begin{enumerate} 
\item Consider a one-sorted variant of the very popular SMT-LIB theory of bit-vectors, where the single domain includes all bit-vectors of the same width.
The set $X$ of all possible sizes of models is comprised of all the powers of two, a decidable set.
Thus, this variant is $X$-decidable, and can be combined with any $X$-shiny theory, thanks to \Cref{theo:Ipart-Ss-Sd}.
\item The first-order theory of fields only has finite models of cardinalities in the set $Y=\{p^{n} : n\in \No, \text{$p$ is prime}\}$, which is again decidable, and is therefore $Y$-decidable.
\item Given a finite field of cardinality $p^{n}$, the theory of vector spaces over this field only has finite models of cardinality in $Z=\{kp^{n} : k\in\mathbb{N}\}$, also decidable, and is therefore $Z$-decidable.
\end{enumerate}
\end{example}

\section{Second Generalization}\label{sec:newmethod}
We now present our second generalization of the shiny combination theorem. It allows one to combine two sets of theories: the first, weaker set, has no requirement on the existence of algorithms for it (other than decidability, of course). The second requires an algorithm for computing a variant of a spectra.

\begin{definition}\label{def:restricted-final}
    Take a computable set $S\subseteq \No$.
    \begin{enumerate}
        \item A decidable theory $\T$ is $S$-restricted if, for every quantifier-free formula $\phi$, $\spec_{\T}(\phi)\setminus S\neq\emptyset$ implies $\aleph_{0}\in\spec_{\T}(\phi)$.
        \item A decidable theory is $S$-final if there is an algorithm that takes a quantifier-free formula $\phi$ and returns $0$ if $\phi$ is not $\T$-satisfiable, and otherwise an $n\in\No$, not necessarily in $S$, such that $\spec_{\T}(\phi)\cap \omg{S}=[n,\aleph_{0}]\cap \omg{S}$.
    \end{enumerate}
\end{definition}

Let us first see relations between these properties and others.

\begin{proposition}\label{prop:limits-Sf-Sr}
  \begin{enumerate}
      \item $\class{S-f}\subseteq\class{SI}\subseteq\class{S-r}$, for any $S$.
      \item $\class{$\No$-r}=\class{}$, $\class{$\No$-f}=\class{shiny}$, and  $\class{$\emptyset$-r}=\class{$\emptyset$-f}=\class{SI}$.
  \end{enumerate}  
\end{proposition}

\begin{proof}
\begin{enumerate} 
\item Suppose $\T$ is $S$-final, and take a quantifier-free $\phi$:
    if $\phi$ is not $\T$-satisfiable we have nothing to prove;
    otherwise our algorithm gives an $n\in\No$ such that $\spec_{\T}(\phi)\cap \omg{S}=[n,\aleph_{0}]\cap \omg{S}$, and since $\aleph_{0}\in \omg{S}$, $\T$ is stably infinite.
    If $\T$ is stably infinite we get that if $\spec_{\T}(\phi)$ is not empty then $\aleph_{0}\in\spec_{\T}(\phi)$, meaning $\T$ is trivially $S$-restricted.

\item An $\No$-restricted theory imposes no further requirements on the theory other than being decidable, as we cannot have $\spec_{\T}(\phi)\setminus \No$ not empty unless $\spec_{\T}(\phi)$ already contains $\aleph_{0}$.
    $\No$-finality implies there is an algorithm that takes a quantifier-free $\phi$ and returns $0$ if $\phi$ is not $\T$-satisfiable, and otherwise an $n\in\No$ such that $\spec_{\T}(\phi)=[n,\aleph_{0}]$, exactly the notion of shininess in \Cref{prop:alternative-definition-shininess}.\lncsqed
    \end{enumerate}\let\qed\relax
\end{proof}

\begin{figure}[t]
\begin{algorithm}[H]
    \begin{algorithmic}[1]
        \Function{SR+SF}{$\phi_1, \phi_2$}
        \If{$\phi_2$ is not $\T_2$-satisfiable}
              \Return false
            \EndIf
          \State{$n\gets$ computed for $\phi_2$ and $\T_2$, according to \Cref{def:restricted-final}.}
            \If{$\phi_1\wedge \neq(x_{1},\ldots,x_{n})$ is $\T_1$-satisfiable}
                \Comment{$x_{1},\ldots,x_{n}$ are fresh}
                \State \Return true
            \EndIf
            \State \Return false
        \EndFunction
\end{algorithmic}
    \caption{Pseudocode for determining non-emptiness of $\spec_{\T_1}(\phi_1)\cap\spec_{\T_2}(\phi_2)$ when $\T_1$ is $S$-restricted and $\T_2$ is $S$-final.} \label{alg-srestricted-sfinal}
\end{algorithm}
\end{figure}

Now, we present our second new combination theorem:
its proof follows that of \Cref{theo:Ipart-Ss-Sd}, but this time using \Cref{alg-sdec-sshiny}.

\begin{restatable}{theorem}{IpartSfSr}\label{theo:Ipart-Sf-Sr}
    Let $\T_{1}$ be $S$-restricted, and $\T_{2}$ be $S$-final.
    Then $\T_{1}\sqcup\T_{2}$ is decidable. 
\end{restatable}

\begin{proof}
        Let $\phi_{1}$ and $\phi_{2}$ be conjunctions of literals over the signatures $\Sigma_{1}$ and $\Sigma_{2}$, respectively.
    We state \Cref{alg-srestricted-sfinal} presents an adequate decision procedure.
    That we can compute an $n\in\No$ such that $\spec_{\T_{2}}(\phi_{2})\cap \omg{S}=[n,\aleph_{0}]\cap \omg{S}$ follows from the fact that $\T_{2}$ is $S$-final.
    We can compute if $\phi_{1}\wedge\neq(x_{1},\ldots,x_{n})$ is $\T_{1}$-satisfiable since $\T_{1}$ is decidable.
    Next we prove correctness.

    ($\Rightarrow$):
    Suppose $\spec_{\T_{1}}(\phi_{1})\cap \spec_{\T_{2}}(\phi_{2})$ is not empty.
    If $\aleph_{0}$ is an element they share then $\phi_{1}\wedge\neq(x_{1},\ldots,x_{n})$ is $\T_{1}$-satisfiable, so we will get true.
    If $\aleph_{0}$ is not in the intersection of the spectra, then $\aleph_{0}\notin \spec_{\T_{1}}(\phi_{1})$ and so $\spec_{\T_{1}}(\phi_{1})\subseteq S$, and the fact that $\phi_{1}\wedge\neq(x_{1},\ldots,x_{n})$ is $\T_{1}$-satisfiable implies $\spec_{\T_{1}}(\phi_{1})\cap[n,\aleph_{0}]\subseteq \omg{S}\cap[n,\aleph_{0}]$ is not empty, thus again we get true.

    ($\Leftarrow$)
    Suppose this time $\spec_{\T_{1}}(\phi_{1})\cap \spec_{\T_{2}}(\phi_{2})$ is empty.
    Then we have that $\phi_{1}\wedge\neq(x_{1},\ldots,x_{n})$ will not be $\T_{1}$-satisfiable, since otherwise the intersection would not be empty.
    As $\aleph_{0}\notin\spec_{\T_{1}}(\phi_{1})$ we get $\spec_{\T_{1}}(\phi_{1})\subseteq S$, and since $\spec_{\T_{2}}(\phi_{2})\cap \omg{S}=[n,\aleph_{0}]\cap \omg{S}$ we get false.

\end{proof}

\begin{example}
The first example of a finitely witnessable (property related to politeness) theory that is not strongly finitely witnessable (property related to strong politeness, equivalent to shininess over decidable theories) was $\Teven$:
this theory has only interpretations of even size, or infinite ones.
Yet, in a sense, it is a surprisingly well-behaved theory:
if $S$ is the set of positive even numbers, then it should not be too hard to prove that $\spec_{\Teven}(\phi)=[n,\aleph_{0}]\cap \omg{S}$, where 
$n$ is the minimal size of any first-order model of $\phi$, which we assume to be quantifier-free.
This is almost the definition of shininess, except for the pesky intersection with $\omg{S}$.
Using the notions just defined, $\Teven$ is seen in fact to be $S$-shiny, and thus combinable with any $S$-restricted theory.
\end{example}

As we did for $S$-decidability and $S$-shininess, we analyze how the sets of $S$-restricted and $S$-final theories change as $S$ changes.


\begin{restatable}{proposition}{incomparableSfSr}\label{prop:incomparable-Sf-Sr}
\begin{enumerate}
    \item If $S\subsetneq T$, $\class{S-r}\subsetneq \class{T-r}$, and $\class{T-f}\subsetneq \class{S-f}$.
    \item If $S$ and $T$ are incomparable, so are $\class{S-r}$ and $\class{T-r}$, and $\class{S-f}$ and $\class{T-f}$.
\end{enumerate}
\end{restatable}

\begin{proof}
        Suppose $\T$ is $S$-restricted, so if $\spec_{\T}(\phi)\setminus S\neq \emptyset$ then $\aleph_{0}\in\spec_{\T}(\phi)$ (for quantifier-free $\phi$).
    If $\spec_{\T}(\phi)\setminus T\neq\emptyset$, as $S\subseteq T$ we obtain $\spec_{\T}(\phi)\setminus S\neq\emptyset$ and thus $\aleph_{0}\in\spec_{\T}(\phi)$, meaning $\T$ is $T$-restricted.

    Now suppose $\T$ is $T$-final:
    there is an algorithm that takes a (quantifier-free) $\phi$ and returns $0$ if it is not $\T$-satisfiable, and otherwise an $n\in\No$ such that $\spec_{\T}(\phi)\cap \omg{T}=[n,\aleph_{0}]\cap \omg{T}$.
    In this last case, since, again, $S\subseteq T$, $T\cap S=S$, and so $\spec_{\T}(\phi)\cap \omg{S}=(\spec_{\T}(\phi)\cap \omg{S})\cap \omg{T}=([n,\aleph_{0}]\cap \omg{T})\cap \omg{S}=[n,\aleph_{0}]\cap \omg{S}$, we get the same algorithm shows $\T$ is $S$-final.
    
    Take $n_{s}\in S\setminus T$, and $n_{t}\in T\setminus S$.
    The theory $\T_{=n_{s}}$ is decidable and $S$-restricted, but not $T$-restricted, and analogous statements can be made about $\T_{=n_{t}}$.

    Define the $\Spn$-theory $\T_{s}$ axiomatized by
    \[\{\psi_{=n_{t}}\rightarrow\neg P_{k} : k\notin U\}\cup\{P_{i}\rightarrow\neg P_{j} : i\neq j\}.\]
    It is decidable:
    given a conjunction of literals $\phi$with at most one positive $P$-literal $P_{k}$, where $\phi^{\prime}$ is obtained from $\phi$ by removing all $P$-literals, we state $\phi$ is $\T_{s}$-satisfiable if and only if $\phi^{\prime}$ is $\T_{Eq}$-satisfiable;
    furthermore, its spectrum equals $[\minmod_{\T_{s}}(\phi^{\prime}),\aleph_{0}]$ if $k\in U$, and $[\minmod_{\T_{s}}(\phi^{\prime}),\aleph_{0}]\setminus\{n_{t}\}$ otherwise.
    Indeed, given a $\T_{Eq}$-interpretation satisfying $\phi^{\prime}$ with any cardinality greater than or equal to $\minmod_{\T_{Eq}}(\phi^{\prime})$, we simply make $P_{n}$ true if and only if $n=k$;
    this construction only fails to work when $k\notin U$ and the desired cardinality is $n_{t}$.

    Because $n_{t}$ is not in $S$, $\spec_{\T_{s}}(\phi)\cap \omg{S}=[\minmod_{\T_{Eq}}(\phi^{\prime}),\aleph_{0}]\cap \omg{S}$ (assuming $\phi^{\prime}$ is $\T_{Eq}$-satisfiable, otherwise we just return $0$), and thus $\T_{s}$ is $S$-final.
    But either $n_{t}$ is not the first element of $T$, and $\spec_{\T_{s}}(P_{k})\cap \omg{T}=\omg{T}\setminus \{n_{t}\}$ for a $k\notin U$ (which is not a set of the form $[n,\aleph_{0}]\cap \omg{T}$);
    or it is, and then $\spec_{\T_{s}}(P_{k})\cap \omg{T}=[n_{t},\aleph_{0}]\cap \omg{T}$ if $k\in U$, and $[n,\aleph_{0}]\cap \omg{T}$ otherwise, for any $n>n_{t}$, meaning no algorithm could output the $n_{t}$, respectively $n$, given $k$.
    $\T_{s}$ is therefore not $T$-final, a similar construction giving a $T$-final theory which is not $S$-final.

    Finally, if $S\subsetneq T$ and $n_{t}\in T\setminus S$, we have that the theories $\T_{=n_{t}}$ and $\T_{t}$ also show $\class{S-r}\subsetneq\class{T-r}$ and $\class{T-f}\subsetneq\class{T-f}$.
\end{proof}

\begin{figure}[htp]
    \centering
    \begin{tikzcd}[sep=0.4mm]
      & \phantom{0} \arrow[bend right=90,swap,red,<->]{dddddd}{\Gal} & & \class{}\arrow[dash]{dl}{}\arrow[dash]{dr}{} & & & & & &\\
      & & \class{S-d}\arrow[dash]{dl}{} & & \class{S-r}\arrow[dash]{ddrr}{} & & & &\\
      & \class{CFS}\arrow[dash]{dd}{} & & & & & & &\\
      \phantom{O}\arrow[dashed,dash]{r}{}&\phantom{O}\arrow[dashed,dash]{rrrrr}{}& & & & & \class{SI}\arrow[dash]{ddll}{}\arrow[dash,dashed]{r} &\phantom{0}\\
      & \class{\G}\arrow[dash]{dr}{} & & & & & & & \\
      & & \class{S-s}\arrow[dash]{dr}{} & & \class{S-f}\arrow[dash]{dl}{} & & & & \\
      & \phantom{0} & & \class{shiny} & & & & & & 
    \end{tikzcd}
    \caption{The general picture involving the notions defined in this paper:
    notice how \Cref{fig:diagram-Ss-Sd} is a more detailed version of the left side of the diagram.}
    \label{fig:full-Hasse-diagram}
  \end{figure}
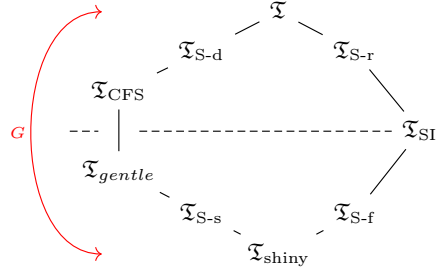

\section{Sharpness of the new methods}\label{sec:test}

We now set out to prove that the new combination theorems \Cref{theo:Ipart-Ss-Sd,theo:Ipart-Sf-Sr} are sharp.
For the former, this means every theory combinable with all $S$-decidable theories must be $S$-shiny, and that 
every theory combinable with all $S$-shiny theories must be $S$-decidable.
Similar sharpness results are proven for the latter.
We utilize the method of {\em test theories}, introduced in \cite{POPL}.

\subsection{Test Theories}
\label{subsec:test}
\begin{table}[t]
\centering
\renewcommand{\arraystretch}{1.5}
\centering
\begin{tabular}{cc}\hline
Theory & Details \\
\hline
$\T_{n}^{m}$ & 
\makecell{
\underline{Axiomatization:} $\{\psi_{=n}\vee\psi_{=m}\}\cup\{P_{k}\rightarrow\psi_{=m} : k\in \unc\}\cup P_{\neq}$\\[.5em]
\underline{Properties:}
Decidable, $S$-decidable (if $n\notin S$), $S$-restricted (if $m,n\in S$)\\
$m\notin\spec_{\T}(\phi), n\in\spec_{\T}(\phi)\Rightarrow [\phi\wedge P_k \text{ is } \T\sqcup\T_{n}^{m}\text{-SAT}\Leftrightarrow k\notin U ] $} \\\arrayrulecolor{gray!50}\hline
$\T_{=}^{P}$ & 
\makecell{
\underline{Axiomatization:}
$\{P_{n}\rightarrow\psi_{=n} : n\in\No\}$ \\[.5em]
\underline{Properties:}
Decidable, $S$-decidable (for any $S$)\\
  $\phi\wedge P_n\text{ is } 
\T\sqcup\T_{=}^{P}\text{-SAT}\Leftrightarrow 
n\in\spec_{\T}(\phi)$}
\\\arrayrulecolor{gray!50}\hline
$\T^{S}_{\leq}$ & 
\makecell{
\underline{Axiomatization:}
$\{P_{n}\rightarrow\psi_{\leq F(n)} : F(n)\in\mathbb{N}\}\cup\{\neg\psi_{=n} : n\notin S\}\cup P_{\neq}$\\[.5em]
\underline{Properties:}
Decidable, $\overline{S}$-decidable, $S$-restricted  \\ 
$\spec_{\T}(\phi)\neq\emptyset, S\cap\spec_{\T}(\phi)=\emptyset\Rightarrow[\phi\wedge P_{k}\text{ is }\T\sqcup\T_{\leq}^{S}\text{-SAT}\Leftrightarrow F(k)=\aleph_{0}]$}
\\\arrayrulecolor{black}\hline
$\T^{S}_{=}$ & 
\makecell{
\underline{Axiomatization:}
$\{P_{n}\rightarrow \psi_{=n} : n\in S\}\cup \{\neg P_{n} : n\notin S\}\cup P_{\neq}$\\[.5em]
\underline{Properties:}
Decidable, $S$-shiny, $S$-restricted \\
$k\in S\Rightarrow[\phi\wedge P_k \text{ is } \T\sqcup\T_{=}^{S}\text{-SAT}\Leftrightarrow k\in\spec_{\T}(\phi)]$
}
 \\\arrayrulecolor{gray!50}\hline
$\T_{S-r}^{m,n}$ & 
\makecell{
\underline{Axiomatization:}
$\{\neg \psi_{=k} : k\notin S\cup \{n\}\}\cup \{\psi_{=n}\vee\psi_{\geq m}\}\cup\{\psi_{=n}\rightarrow\neg P_{k} : k\notin U\}\cup P_{\neq}$\\[.5em]
\underline{Properties:}
 Decidable, $S$-final \\ 
$n\in\spec_{\T}(\phi), m=M_{\T}(\phi)\Rightarrow[\phi\wedge P_{k} \text{ is } \T\sqcup\T^{m,n}_{S-r}\text{-SAT}\Leftrightarrow k\in U]$}
\\\arrayrulecolor{black}\hline
\end{tabular}
\vspace{2mm}
\caption{Test theories and their properties; 
$S\subseteq \No$;
for the first and last theories, $m>n$;
for the last, $n\notin S$. 
We define $P_{\neq}=\{P_{i}\rightarrow\neg P_{j} : i\neq j\}$ and $M_{\T}(\phi)=1+\max\spec_{\T}(\phi)$ to make the table more readable.
}
\label{tab:othertesttheories}
\end{table}

The idea behind test theories goes more or less as follows:
we want to prove that if a theory $\T$ is combinable with every theory of $X$, then it admits some property;
this could be done by looking at all $\T\sqcup\T^{\prime}$, for $\T^{\prime}\in X$.
But considering all theories of $X$ is rarely feasible, and is certainly not very effective.
So, from $X$, we select special theories, {\em test theories}, such that the satisfiability problem of the combination of $\T$ and a test theory directly relates to properties from $\T$.
In this sense, the set of test theories will be a very difficult set of theories to combine with $\T$, as being combinable with them forces $\T$ to have a given property.

We need $5$ sets of test theories, three found in \cite{POPL}, and two new. 
Let $\Spn$ be the signature with no function symbols,  
and nullary predicate symbols $\{P_{n}:n\in\No\}$:
all of our test theories will be $\Spn$-theories, and
are described in \Cref{tab:othertesttheories}, which relies on
 an undecidable set $U\subseteq \No$, 
 and a function $F:\No\rightarrow\No\cup\{\aleph_{0}\}$ such that $\{(m,n)\in{\No}^{2} : n\leq F(m)\}$ is decidable, but $\{n\in\No : F(n)=\aleph_{0}\}$ is not.\footnote{An example for $F$ assumes an enumeration of Turing machines, takes the index of a machine, and  returns the number of steps it takes to halt, or $\aleph_{0}$ if it doesn't.}

The three test theories from \cite{POPL} are $\T^{m}_{n}$, $\T^{P}_{=}$ and $\T^{S}_{\leq}$.
$\T^{m}_{n}$ is a $\Spn$-theory with only interpretations of size $m$ or $n$, and where $P_{k}$ only has interpretations of size $m$ whenever $k\in U$;
so $\spec_{\T^{m}_{n}}(P_{k})=\{n,m\}$ if $k\notin U$, and $\{m\}$ otherwise.
A $\T^{P}_{=}$-interpretation where $P_{n}$ is true must have exactly $n$ elements;
so $\spec_{\T^{P}_{=}}(P_{n})=\{n\}$.
The theory $\T_{\leq}^{S}$ has only interpretations of cardinality in $\omg{S}$, and if $P_{n}$ is true and $F(n)$ is finite, then we have at most $F(n)$ elements;
so $\spec_{\T}(P_{n})=S\cap[1,F(n)]$ if $F(n)$ is finite, and $\omg{S}$ otherwise.

The two new test theories are $\T^{S}_{=}$ and $\Test{S-r}^{m,n}$.
The theory $\T^{S}_{=}$ is similar to $\T_{=}^{P}$, with a twist:
if $P_{n}$ is true (what can only happen  for $n\in S$) then the interpretation has size $n$.
So, $\spec_{\T^{S}_{=}}(P_{n})=\emptyset$ if $n\notin S\subseteq \No$, and $\{n\}$ otherwise.
The theory $\Test{S-r}^{m,n}$ is quite similar to $\T^{m}_{n}$, and will be used to test whether the spectrum of a formula is entirely contained in a computable, infinite set $S\subseteq \No$.
It depends on two numbers $m$ and $n$, with $m>n$ and $n\notin S$, and its interpretations must:
have cardinality in $S$ or equal to $n$;
have cardinality $n$ or at least $m$;
and, if they have cardinality $n$, then $P_{k}$ can only be true if $k\in U$.
To summarize it, $\spec_{\Test{S-r}^{m,n}}(P_{k})=\omg{S}\cap(\{n\}\cup[m,\aleph_{0}])$ if $k\in U$, and $\omg{S}\cap[m,\aleph_{0}]$ otherwise.

\begin{restatable}{theorem}{anothertableiscorrect}    
\label{theo:other-test-theories}
\label{theo:test-theories}
    \Cref{tab:othertesttheories} is correct:
    each theory satisfies the properties listed for it.
\end{restatable}

\begin{proof}
\begin{enumerate}
\item $\T^{m}_{n}$ is decidable according to \cite{POPL}.
So we move on to prove that, if $n\notin S$, $\T^{m}_{n}$ is $S$-decidable, and if $m,n\in S$, then $\T^{m}_{n}$ is $S$-restricted.
$\T_{n}^{m}$ only has interpretations of cardinality $n$ or $m$, so the $S$-decidability of $\T_{n}^{m}$ will follow from the fact every $\T_{n}^{m}$-satisfiable quantifier-free formula $\phi$ has a $\T_{n}^{m}$-interpretation of size $m$:
indeed, it is enough to simply add $m-n$ more elements to the domain of a $\T_{n}^{m}$-interpretation of size $n$ that satisfies $\phi$.
It is obvious that $\T^{m}_{n}$ is $S$-restricted as the spectra of all its formulas are contained in $\{n,m\}$, itself contained in $S$.

\item $\T^{P}_{=}$ is decidable by \cite{POPL}.
Let $\phi$ be a conjunction of literals, without loss of generality with at most one positive $P$-literal, and let $\phi^{\prime}$ be the result of removing from $\phi$ any $P$-literals.
If the positive $P$-literal in $\phi$ is $P_{n}$, then $\spec_{\T^{P}_{=}}(\phi)=\{n\}$ if $\minmod_{\T_{Eq}}(\phi^{\prime})\leq n$, and is empty otherwise.
If $\phi$ contains no positive $P$-literals then $\spec_{\T^{P}_{=}}(\phi)=[\minmod_{\T_{Eq}}(\phi^{\prime}),\aleph_{0}]$.
In both cases we can decide whether an element $m$ of $S$ is in these sets, so $\T^{P}_{=}$ is $S$-decidable.

\item Again we find the proof that $\T_{\leq}^{S}$ is decidable in \cite{POPL}.
Now we provide a proof that $\T_{\leq}^{S}$ is $\overline{S}$-decidable and $S$-restricted.
It is obvious this theory is $\overline{S}$-decidable, as it contains no interpretations of cardinality in $\overline{S}$;
it is also obvious this theory is $S$-restricted, as it only contains interpretations of cardinality in $\omg{S}$.

\item Let $\phi$ be a conjunction of literals, with at most one positive $P$-literal (otherwise the formula is immediately not $\T^{S}_{=}$-satisfiable), and let $\phi^{\prime}$ be the result of removing from $\phi$ all of its $P$-literals.
If the positive $P$-literal in $\phi$ is $P_{n}$, for $n\notin S$, then $\phi$ is not $\T^{S}_{=}$-satisfiable;
if $n\in S$, then $\spec_{\T^{S}_{=}}(\phi)=\{n\}$ if $\minmod_{\T_{Eq}}(\phi^{\prime})\leq n$, and is empty otherwise.
For the non-trivial part of this statement, when $\minmod_{\T_{Eq}}(\phi^{\prime})\leq n$, we take a $\T_{Eq}$-interpretation satisfying $\phi^{\prime}$ with $n$ elements (what we can do as $\T_{Eq}$ is smooth), make $P_{n}$ true and all $P_{m}$ false.

If $\phi$ contains no positive $P$-literals then $\spec_{\T^{S}_{=}}(\phi)=[\minmod_{\T_{Eq}}(\phi^{\prime}),\aleph_{0}]$.
This proves not only that this theory is decidable, but also that it is $S$-shiny:
in the case that the spectrum equals $\{n\}$ we must have $n\in S$.

Suppose now that $k\in\spec_{\T}(\phi)$, so there is a $\T$-interpretation $\A$ that satisfies $\phi$ with $k$ elements in its domain;
we make it into a $\T\sqcup\T^{S}_{=}$-interpretation by making $P_{k}$ true and all other $P_{i}$ false, and of course it satisfies $\phi\wedge P_{k}$.
Reciprocally, suppose $\phi\wedge P_{k}$ is satisfied by the $\T\sqcup\T_{=}^{S}$-interpretation $\A$:
because $\A$ is a $\T_{=}^{S}$-interpretation satisfying $P_{k}$ it contains $k$ elements;
because $\A$ is also a $\T$-interpretation satisfying $\phi$, we get $k\in\spec_{\T}(\phi)$.

It is obvious $\T^{S}_{=}$ is $S$-restricted, as it only has interpretations of cardinality in $S$.

\item We proceed to prove $\Test{S-r}^{m,n}$ is decidable and, if $n\notin S$ and $m\geq n$, $S$-final.
Take a conjunction of literals $\phi$:
        we may assume it contains at most one positive $P$-literal, otherwise it will not be $\Test{S-r}^{m,n}$-satisfiable.
        Let $\phi^{\prime}$ be the result of taking from $\phi$ any $P$-literals, and we state $\phi$ is $\Test{S-r}^{m,n}$-satisfiable if and only if $\phi^{\prime}$ is $\T_{Eq}$-satisfiable.
        
        For the non-trivial direction
        take a $\T_{Eq}$-interpretation $\A$ that satisfies $\phi^{\prime}$ with any cardinality in $\omg{S}\cap [m_{0},\aleph_{0}]$, for $m_{0}=\max\{m,\minmod_{\T_{Eq}}(\phi^{\prime})\}$ (so the set does not contain $n$ and is not empty), make $P_{k}$ true only if it occurs positively in $\phi$, and the resulting interpretation is a $\Test{S-r}^{m,n}$-interpretation that satisfies $\phi$.
        Moreover, no element in $\omg{S}\cap[1,m_{0}-1]$ can be in the spectrum of $\phi$, as either:
        $\phi^{\prime}$ is not satisfiable, if $m_{0}=\minmod_{\T_{Eq}}(\phi^{\prime})$;
        or simply $S\cap[1,m_{0}-1]$ is empty, if $m_{0}=m$.

    Now, let $n\in\spec_{\T}(\phi)\setminus S$ and $m=1+\max(\spec_{\T}(\phi)\cup\{n\})$, and suppose that $k\in U$:
    because $n\in\spec_{\T}(\phi)$ there is a $\T$-interpretation $\A$ that satisfies $\phi$ with $n$ elements in its domain;
    we make $P_{k}$ true and all other $P_{i}$ false in $\A$, and the result is a $\T\sqcup\T^{m,n}_{S-r}$-interpretation as $k\in U$, that in addition satisfies $\phi\wedge P_k$.
    Reciprocally, assume $\A$ is a $\T\sqcup\T^{m,n}_{S-r}$-interpretation that satisfies $\phi\wedge P_{k}$:
    because $\A$ is a $\T$-interpretation that satisfies $\phi$, the cardinality of its domain is bounded by $\max(\spec_{\T}(\phi))$;
    since $\A$ is also a $\T^{m,n}_{S-r}$-interpretation this implies the domain of $\A$ has exactly $n$ elements, what in turn implies $k\in U$.

\end{enumerate}
\end{proof}

\subsection{Sharpness for $S$-shininess}

We now prove that \Cref{theo:Ipart-Ss-Sd} is sharp, that is:
Given $S$,
all decidable theories that can be combined with all $S$-shiny theories
must be $S$-decidable,
and all decidable theories that can be combined with all $S$-decidable theories
must be $S$-shiny.
Formally, we utilize operator $G$ from \Cref{sec:theorycombprops}, and the test theories $\T^{S}_{=}$ to prove $G(\class{S-s})\subseteq\class{S-d}$, and $\T_{=}^{P}$, $\T_{\leq}^{S}$ and $\T_{n}^{m}$ to prove $G(\class{S-d})\subseteq \class{S-s}$.

\begin{restatable}{theorem}{IIpartSsSd}\label{theo:IIpart-Ss-Sd}
$G(\class{S-s})\subseteq \class{S-d}$ and $G(\class{S-d})\subseteq \class{S-s}$.
\end{restatable}

\begin{proof}
    \begin{enumerate}
        \item Suppose that $\T\in G(\class{S-s})$, and take a quantifier-free formula $\phi$:
        since $\T^{S}_{=}$ is decidable and $S$-shiny (see \Cref{theo:test-theories}), for any $n\in S$ we have $n\in\spec_{\T}(\phi)$ if and only if $\phi\wedge P_{n}$ is $\T\sqcup\T^{S}_{=}$-satisfiable, a problem for which we have an algorithm since $\T\sqcup\T^{S}_{=}$ is decidable;
        we thus have $\T$ is $S$-decidable.

\begin{figure}[t]
        \begin{algorithm}[H]
    \begin{algorithmic}[1]
    \Function{SS}{$\phi$}
        \State{$k\gets 1$}
        \State{$S_{0}\gets\emptyset$}
        \While{True}
        \If{$k\in\spec_{\T}(\phi)$}
            \If{$k\notin S$}
                \Return $(k,S_{0})$
            \EndIf
            \State{$S_{0}\gets S_{0}\cup\{k\}$}
            \State{$k\gets k+1$} 
        \Else    
            \If{$\phi\wedge\neq(x_{1},\ldots,x_{k})$ is not $\T$-satisfiable}
                \Return $(0,S_{0})$
            \EndIf
            \State{$k\gets k+1$}
        \EndIf
        \EndWhile
           \EndFunction 
\end{algorithmic}
    \caption{Pseudocode for showing $\T$ is $S$-shiny if $\T\in G(\class{S-d})$.} \label{alg:S-shiny-proof}
\end{algorithm}
\end{figure}

        \item Suppose now that $\T\in G(\class{S-d})$.
        First we show that $\T$ must have computable finite spectra.
        $\Testcfs$ is decidable and $S$-decidable by \Cref{theo:other-test-theories};
        this means $\T\sqcup\Testcfs$ is decidable, and $k\in\spec_{\T}(\phi)$ (for a quantifier-free $\phi$) if and only if $\phi\wedge P_{k}$ is $\T\sqcup\Testcfs$-satisfiable, proving $\T$ indeed has computable finite spectra.

We now prove $\T$ is $S$-shiny by showing, in \Cref{alg:S-shiny-proof}, the algorithm that outputs $(n,S_{0})$.
Checking (for $k\in \No$) whether $k\in S$ in line $5$ is possible because $S$ is computable, whether $k\in \spec_{\T}(\phi)$ in line $4$ is possible because we just proved $\T$ has computable finite spectra, and whether $\phi\wedge\neq(x_{1},\ldots,x_{k})$ is $\T$-satisfiable in line $10$ is possible because $\T$ is decidable.
We notice that $S_{0}$ is nothing but $\spec_{\T}(\phi)\cap[1,k-1]$.
Now, to see that our algorithm stops, there are two cases to consider:
whether $\spec_{\T}(\phi)\setminus S$ is empty or not.
        
If it is empty, and $\spec_{\T}(\phi)$ is infinite, using \Cref{compactness2} we get $\aleph_{0}\in\spec_{\T}(\phi)$, and that leads to a contradiction:
$\spec_{\T}(\phi)\setminus S=\emptyset$, $\aleph_{0}\notin S$ and $\aleph_{0}\in \spec_{\T}(\phi)$ cannot be simultaneously true.
So, if $\spec_{\T}(\phi)\setminus S$ is empty, $\spec_{\T}(\phi)$ must be finite (and contained in $S$), and eventually we will find a $k$ such that $k\notin\spec_{\T}(\phi)$ and $\phi\wedge\neq(x_{1},\ldots,x_{k})$ is not $\T$-satisfiable, when our algorithm must stop.
In that case it returns $(0,S_{0})$ for $S_{0}=\spec_{T}(\phi)\cap[1,k-1]$, which is a finite subset of $S$ and equals $\spec_{\T}(\phi)$.
        
So now we assume $\spec_{\T}(\phi)\setminus S\neq\emptyset$:
we will then prove $\spec_{\T}(\phi)$ has a finite value not in $S$, and our algorithm will stop when it finds the first such value, and will output $(k,S_{0})$ for $S_{0}=\spec_{\T}(\phi)\cap[1,k-1]$ a finite subset of $S$;
second, we will prove that the existence of a finite value $k$ in $\spec_{\T}(\phi)$ that is not in $S$ implies $[k,\aleph_{0}]\subseteq \spec_{\T}(\phi)$, so the spectrum indeed has the form suggested by the algorithm as $\spec_{\T}(\phi)$ equals the union of $\spec_{\T}(\phi)\cap[1,k-1]$ and $[k,\aleph_{0}]$.
        
Suppose, then, that $\spec_{\T}(\phi)\setminus S\neq\emptyset$ but $(\spec_{\T}(\phi)\setminus S)\cap \mathbb{N}=\emptyset$.
We then have that $\spec_{\T}(\phi)\setminus S=\{\aleph_{0}\}$, and therefore $\phi\wedge P_{i}$ is $\T\sqcup\Tle^{\overline{S}}$-satisfiable if and only if $F(i)=\aleph_{0}$, an undecidable problem.
This leads to a contradiction, as $\Tle^{\overline{S}}$, for $\overline{S}=\No\setminus S$, is decidable and has no interpretations with sizes in $S$, thus being $S$-decidable (see \Cref{theo:test-theories}). 
Next, suppose $n\in \spec_{\T}(\phi)\setminus S$ is finite and there is an $n<m<\aleph_{0}$ such that $m\notin\spec_{\T}(\phi)$.
Then $\phi\wedge P_{i}$ is $\T\sqcup\T_{n}^{m}$-satisfiable iff $i\notin U$, an undecidable problem.
This too leads to a contradiction, as $\T_{n}^{m}$ is decidable and $S$-decidable (\Cref{theo:test-theories}).
\lncsqed
    \end{enumerate}\let\qed\relax
\end{proof}

\subsection{Sharpness for $S$-finality}
We now prove that \Cref{theo:Ipart-Sf-Sr} is sharp, namely:
given $S$,
every decidable theory that can be combined with all $S$-final theories must be $S$-restricted,
and every decidable theory that can be combined with all $S$-restricted theories must be $S$-final. Again, we utilize operator $G$ from \Cref{sec:prem}, the test theory $\Test{S-r}^{m,n}$ to prove $G(\class{S-f})\subseteq \class{S-r}$, and $\T_{\leq}^{S}$, $\T_{=}^{S}$ and $\T_{n}^{m}$ to prove $G(\class{S-r})\subseteq \class{S-f}$.

\begin{restatable}{theorem}{IIpartSfSr}\label{theo:IIpart-Sf-Sr}
\label{prop:connectionsconn}

    $G(\class{S-f})\subseteq \class{S-r}$ and $G(\class{S-r})\subseteq \class{S-f}$.
\end{restatable}

\begin{proof}
    \begin{enumerate}
        \item Suppose that $\T\in G(\class{S-f})$ is not $S$-restricted, that is, there are a quantifier-free formula $\phi$, an $n\in \spec_{\T}(\phi)\setminus S$ and a maximum element of $\spec_{\T}(\phi)$;
        define $m=1+\max\spec_{\T}(\phi)$.
        We state that $\Test{S-r}^{m,n}$ is decidable and $S$-final, so $\T\sqcup\Test{S-r}^{m,n}$ should be decidable;
        yet $\phi\wedge P_{k}$ is $\T\sqcup\Test{S-r}^{m,n}$-satisfiable if and only if $k\in U$, leading to a contradiction.

        \item Suppose that $\T\in G(\class{S-r})$:
        first we prove that, if $\phi$ is $\T$-satisfiable (otherwise we just return $0$), there is a finite number in $\spec_{\T}(\phi)\cap \omg{S}$, for any quantifier-free formula $\phi$.
        Were that not the case we would get $\phi\wedge P_{k}$ is $\T\sqcup\T_{\leq}^{S}$ if and only if $F(k)=\aleph_{0}$, what contradicts the fact that $\T_{\leq}^{S}$ is $S$-restricted.

        Next we show how the least element of $\spec_{\T}(\phi)\cap \omg{S}$ can be calculated:
        using that $\T^{S}_{=}$ is decidable and $S$-restricted, we simply search for the first $k\in S$ such that $\phi\wedge P_{k}$ is $\T\sqcup\T^{S}_{=}$-satisfiable, which must be finite by what we have seen above.
        Finally, suppose $n\in \spec_{\T}(\phi)\cap S$ but, for some $n<m<\aleph_{0}$ in $S$ (which we can always find as $S$ is infinite), $m\notin \spec_{\T}(\phi)$.
        Then $\phi\wedge P_{k}$ is $\T\sqcup\T_{n}^{m}$-satisfiable if and only if $k\notin U$, contradicting the fact that $\T_{n}^{m}$ is decidable and $S$-restricted. \lncsqed
    \end{enumerate}\let\qed\relax
    \end{proof}

\section{Conclusion and Future Work}
\label{sec:conclusion}

We have presented two infinite families of sharp combination methods resembling shininess.
One motivates gentleness as simply one extreme of a range of properties, the other extreme being shininess itself.
The other requires very little from the weaker theory to be combined, not unlike Nelson--Oppen and polite theory combination.
This is in line with \cite{POPL}, which suggests that new combination methods should also be proven sharp given the new techniques available to researchers in the area.

Directions for future work are many:
the Galois connection of theory combination makes finding new combination methods much easier;
but perhaps more interesting than finding new methods would be trying to generalize those we have at hand to the many-sorted case, a task currently being carried out.

\section*{Acknowledgments}
This research was supported in part by 
Defense Advanced Research Projects Agency (DARPA) contract FA875024-2-1001 and
Israel Science Foundation (ISF) grant number 209/25.

\bibliography{bib}

\end{document}